\documentclass[twoside]{article}

\usepackage[preprint]{aistats2027}
\usepackage[round]{natbib}
\usepackage{comment}
\usepackage{xcolor}
\usepackage[
    colorlinks=true,
    citecolor=blue,
    linkcolor=blue,
    urlcolor=blue
]{hyperref}

\usepackage{amsmath,amssymb,amsthm,mathtools,bm}
\usepackage{float}
\usepackage{graphicx,booktabs,microtype,url}
\newtheorem{theorem}{Theorem}
\newtheorem{proposition}{Proposition}
\newtheorem{corollary}{Corollary}
\newtheorem{remark}{Remark}
\theoremstyle{definition}\newtheorem{definition}{Definition}

\newcommand{\C}{\mathcal C}

\newcommand{\Var}{\operatorname{Var}}
\newcommand{\KL}{\operatorname{KL}}
\newcommand{\elpd}{\operatorname{elpd}}
\newcommand{\loo}{\mathrm{LOO}}
\newcommand{\pc}{\mathrm{PC}}

\begin{document}

%

%

\twocolumn[
\aistatstitle{Predictive-Context LOO: Bayesian Model Comparison with Pretrained 
Predictive Learners}
\aistatsauthor{Hans Montcho \And Maurizio Filippone}
\aistatsaddress{KAUST, Saudi Arabia}
]

\begin{abstract}

Leave-one-out (LOO) cross-validation, a standard tool for Bayesian model 
assessment,  evaluates each observation using a LOO predictive 
distribution derived from the model itself.
In this work, we introduce an alternative construction in which posterior
predictive simulation connects the Bayesian model to an external 
pretrained predictive learner.
Our method, Predictive-Context LOO (PC-LOO), fits a Bayesian model once, 
generates posterior-predictive synthetic data, and uses them as context for 
the learner to approximate each LOO predictive density.
Under an oracle construction, averaging contexts generated from the LOO 
posterior recovers the exact LOO predictive density.
Practical PC-LOO replaces the oracle predictor by the pretrained learner
and the LOO posterior by the full posterior. We quantify these errors 
through learner compatibility $\C_i$ and deletion influence $V_i$, obtaining  
$H^2(p_i^{\loo},q_i^{\pc})\leq \C_i + V_i.$
The  identity also holds for any fixed  subset of $M$ context points,
allowing  $M \ll n-1$ as the learner's context budget,
and scalability to large datasets.
Experiments with TabICL support this decomposition, deletion influence decays
approximately as $n^{-1}$, while learner compatibility can remain limiting 
at large $n$. PC-LOO remains accurate with reduced context sizes,
while model comparison accuracy is governed by differences in 
approximation error across models rather than their absolute errors.
More broadly, PC-LOO establishes posterior predictive simulation as a
principled interface for integrating pretrained predictive learners 
into Bayesian workflows.

\end{abstract}

\section{Introduction}
\label{sec:introduction}

Leave-one-out (LOO) cross-validation is a standard tool for Bayesian 
model comparison 
\citep{geisser1979predictive, vehtari2012survey}.
For observations $D=\{(x_j,y_j)\}_{j=1}^n$, LOO evaluates each model 
through the predictive distributions $p(y_i\mid x_i,D_{-i})$, 
where $D_{-i} = D \setminus \{(x_i, y_i)\}$ denotes the data
with observation $i$ removed. 
Traditionally, both exact LOO and its  approximations \citep{vehtari2024pareto}
construct the LOO predictive distribution from the Bayesian model itself, by 
first correcting posterior deviations in the parameter space,
due to leaving a datum out, and then propagating the uncertainty into the
predictive space. In this work, we propose a new perspective to
Bayesian model assessment by directly targeting the predictive space
through synthetic data. We represent the Bayesian model's predictive 
information through synthetic data and transfer it,  as predictive 
context, to an external pretrained predictive learner to approximate 
LOO predictive  distributions.

More broadly,  this
predictive space view only requires  a synthetic 
predictive context and a corresponding prediction task, and its
construction is readily extensible.
As predictive learners evolve \citep{pmlr-v267-reuter25a, 
pmlr-v267-muller25d} and  become more diverse 
the same framework will accommodate structured, 
dependent, and joint prediction problems, 
including leave-group, block or cluster, causal prediction, 
without redesigning the underlying methodology for each setting;
while importance sampling approaches may require problem specific derivations.

Posterior predictive simulation provides a natural interface for this
transfer by representing a fitted Bayesian model through synthetic data.
Pretrained predictive learners, benefiting from in-context learning \citep{dong2024survey},
can then use such labeled data as context to
form predictive distributions without  refitting. 
Examples include TabPFN and TabICL in tabular data, 
\citep{hollmann2022tabpfn,tabicl},  TimePFN for times series
\citep{taga2025timepfn}, SurvPFN for survival data \citep{bohm2026survpfn}, and
CausalFM in causal inference settings \citep{ma2026foundation}, just to cite a few.

Related simulation based approaches amortize Bayesian model comparison by
training neural networks to estimate model probabilities or Bayes 
factors directly
\citep{radev2021amortized,jeffrey2024evidence}.
Here, instead of training a model-specific comparison network, we use
posterior predictive data to communicate a fitted Bayesian model to an
external pretrained predictive learner, motivating 
Predictive-Context LOO (PC-LOO).

More specifically, PC-LOO fits the Bayesian model once to the full dataset and 
generates posterior predictive replicates 
$\widetilde D^{(s)}=\{(x_j,\tilde y_j^{(s)})\}_{j=1}^n.$ For a target
observation $i$, 
$\widetilde D_i^{(s)}= (X_{-i},\widetilde Y_{-i}^{(s)})$ is supplied
as context to the predictive learner with $x_i$  as  query
point.  Averaging the resulting predictive distributions 
over replicated contexts defines the PC-LOO approximation to
$p(y_i\mid x_i,D_{-i})$. Importantly,  observed responses are only used
to evaluate the  predictive distribution, never as learner context.

This  raises two statistical questions. First, does 
averaging predictions over posterior predictive contexts recover
the Bayesian predictive distribution exactly? 
Second, practical PC-LOO replaces the oracle
conditional predictor by a pretrained learner and the LOO posterior
$p(\theta\mid D_{-i})$ by the full posterior
$p(\theta\mid D)$. How do these substitutions affect the
LOO approximation?

The first question has an exact answer: when synthetic contexts are
generated from the LOO posterior and predictions use the corresponding
Bayesian conditional predictor, averaging over contexts recovers the
exact LOO predictive distribution. Thus, the predictive-context
representation itself introduces no approximation error.
For the second question, we characterize
the two departures of practical PC-LOO from this construction by learner
compatibility $\C_i$ and deletion influence $V_i$, and show that they
jointly control the finite-sample approximation error.

The exact representation also holds for any fixed context subset of size
$M$, separating the Bayesian model's data size $n$ from the learner's
context budget $M$; thus, $M\ll n-1$ introduces no oracle subsampling
bias. In practical PC-LOO, its effect enters through the learner's ability
to recover predictive structure from a smaller context, and  provides
a complementary approach to existing subsampling LOO
methods  for large-data \citep{magnusson2020leave}.

For model comparison, the relevant quantity is differential rather than
absolute approximation error: the error in a PC-LOO score difference is
the difference between the approximation errors of the competing models.
Thus, PC-LOO may preserve model comparisons even when individual LOO
predictions are approximate.

We organize our empirical study around the theory 
rather than providing a broad predictive benchmark.
Controlled Gaussian experiments examine the exact representation,
learner compatibility, deletion influence, context size, and model comparison;
real tabular problems study context scaling, large-$n$ behavior, and
accuracy relative to exact  and approximate LOO.

\noindent\textbf{Our contributions are:}

(1) We introduce a predictive-context representation of Bayesian
prediction through posterior predictive simulation and prove its exactness
under the corresponding oracle predictor.

(2)  We characterize practical PC-LOO through learner compatibility
$\C_i$ and deletion influence $V_i$, derive a finite-sample error bound,
show that $V_i=O_p(n^{-1})$ under regular fixed-dimensional conditions,
and establish consistency when the learner is predictively consistent
as its context grows.

(3) We show that PC-LOO model comparison depends on differential rather
than absolute approximation error, and empirically evaluate the theory
using controlled Gaussian experiments and TabICL on tabular problems.

\begin{figure*}[t]
    \centering
    \includegraphics[width=\textwidth,
    trim={0 20 0 20},
    clip]{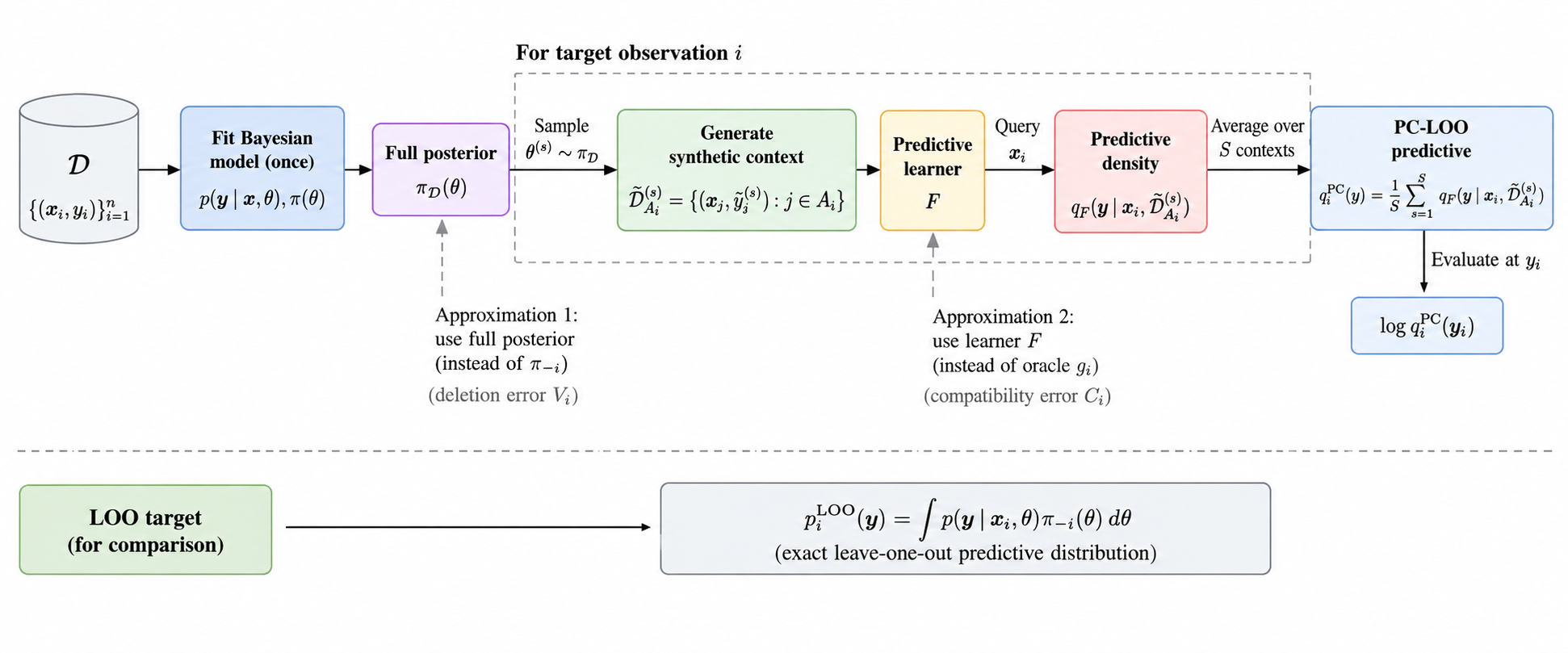}
    \caption{\textbf{PC-LOO workflow.}
    The Bayesian model is fitted once to the full data. For target observation $i$,
    draws from the full posterior generate synthetic response contexts on
    $A_i \subseteq \{1,\ldots,n\}\setminus\{i\}$.
    A predictive learner $F$ predicts at $x_i$ from each context, and the resulting
    densities are averaged and evaluated at $y_i$.
    The approximation to exact LOO arises from using the full rather than LOO
    posterior and from replacing the oracle conditional predictor by $F$.}
    \label{fig:pcloo-workflow}
\end{figure*}

\section{Bayesian LOO and notation}
\label{sec:bayesian-loo}
Let $D=\{(x_j,y_j)\}_{j=1}^n$ denote observations from a Bayesian
conditional model $p(y\mid x,\theta)$ with prior $\pi(\theta)$.
We write
$$
\pi_D(\theta)=p(\theta\mid D), \qquad \pi_{-i}(\theta)=p(\theta\mid D_{-i}),
$$
for the full-data and leave-one-out posteriors, where
$D_{-i}=D\setminus\{(x_i,y_i)\}$. The LOO predictive distribution for
observation $i$ is
$$
p_i^{\loo}(y)=\int p(y\mid x_i,\theta)\,\pi_{-i}(\theta)\,d\theta ,
$$
with pointwise log predictive,  and corresponding  expected log predictive
$$
\ell_i^{\loo}=\log p_i^{\loo}(y_i), \qquad 
\elpd_{\loo}=\sum_{i=1}^n \ell_i^{\loo}.
$$
For  models $\mathcal M_0$ and $\mathcal M_1$, define
$$
\Delta_{\loo} = \elpd_{\loo}(\mathcal M_1) -\elpd_{\loo}(\mathcal M_0).
$$

Exact LOO requires the posteriors $\pi_{-i}$, but 
Pareto Smooth Importance Sampling (PSIS-LOO) \citep{vehtari2024pareto}
instead uses importance sampling from $\pi_D$, with raw ratios
$$
r_i^{(s)} \propto \frac{1}{p(y_i\mid x_i,\theta^{(s)})}
$$
and Pareto-smooths their upper tail; the fitted shape parameter
$\hat{k}_i$ provides a diagnostic of importance-sampling reliability
.

\section{Predictive-context Leave-One-Out}
\label{sec:construction}

In this section, we formalize the predictive-context mechanism 
underlying PC-LOO,  illustrated in Figure~\ref{fig:pcloo-workflow}.
Here, the construction is defined for a generic parameter
distribution and predictive learner, allowing us to separate the exact
representation induced by synthetic contexts from the approximations
introduced in practical PC-LOO.

\subsection{Bayesian predictives and synthetic contexts}

\begin{definition}[Predictive learner]
Let $\mathcal X$ and $\mathcal Y$ denote the covariate and response
spaces, and $\mathcal Q(\mathcal Y)$ the class of densities on
$\mathcal Y$ with respect to a fixed dominating measure.
A predictive learner is a mapping
$$
F:
\mathcal X \times
\bigcup_{m\geq 0}(\mathcal X\times\mathcal Y)^m
\longrightarrow
\mathcal Q(\mathcal Y).
$$
For a query $x$ and finite context $\mathcal D_c$, denote the resulting
predictive density by $q_F(\cdot\mid x,\mathcal D_c)$.
\end{definition}

For any parameter density $\rho(\theta)$, define
$$
p_i^\rho(y)
=
\int p(y\mid x_i,\theta)\rho(\theta)\,d\theta .
$$
In particular, $\rho=\pi_{-i}$ gives the exact LOO predictive
$p_i^{\loo}$, while $\rho=\pi_D$ is available from the full-data fit.

\begin{definition}[Synthetic predictive context]
Fix a target $i$, context indices
$$
A_i\subseteq\{1,\ldots,n\}\setminus\{i\},
\qquad |A_i|=M,
$$
and parameter density $\rho(\theta)$. Generate
$$
\theta^{(s)}\sim\rho(\theta),\qquad
\widetilde Y_j^{(s)}\mid\theta^{(s)}
\sim p(\cdot\mid x_j,\theta^{(s)}),
\qquad j\in A_i,
$$
conditionally independently given $\theta^{(s)}$. The resulting
synthetic predictive context is
$$
\widetilde D_{A_i}^{(s)}
=
\{(x_j,\widetilde Y_j^{(s)}):j\in A_i\}.
$$
\end{definition}

Thus, the learner receives only model-generated responses as context,
never the observed responses in $A_i$.

\subsection{Predictive distributions induced by synthetic contexts}

For fixed $i$, $A_i$, and $\rho$, the construction above induces the
joint density
$$
m_{i,A_i}^{\rho}
\left(y,\widetilde y_{A_i}\right)
=
\int
p(y\mid x_i,\theta)
\prod_{j\in A_i}
p(\widetilde y_j\mid x_j,\theta)
\rho(\theta)\,d\theta,
$$
where
$\widetilde y_{A_i}=\{\widetilde y_j:j\in A_i\}$, with context marginal
$$
m_{A_i}^{\rho}
\left(\widetilde y_{A_i}\right)
=
\int
\prod_{j\in A_i}
p(\widetilde y_j\mid x_j,\theta)
\rho(\theta)\,d\theta .
$$

The conditional predictive density induced by this synthetic-data
experiment is
$$
g_{i,A_i}^{\rho}
\left(y\mid \widetilde y_{A_i}\right)=
\frac{m_{i,A_i}^{\rho}\left(y,\widetilde y_{A_i}\right)}
{m_{A_i}^{\rho}\left(\widetilde y_{A_i}\right)},
$$
whenever $m_{A_i}^{\rho}(\widetilde y_{A_i})>0$. We refer to 
$g_{i,A_i}^{\rho}$ as the oracle predictive-context density.

The synthetic responses in the context provide
information about the common parameter $\theta$. The oracle predictor
updates $\rho$ using this synthetic context and then forms the
Bayesian predictive density at $x_i$. The external learner $F$ 
is asked to solve the same predictive task from $\widetilde D_{A_i}$ alone,
without direct access to $\rho$, $\theta$, or the Bayesian model.

\subsection{Exact predictive-context identity}

Conditioning on a synthetic context generally changes the predictive
distribution at $x_i$, because the context carries information about the 
parameter draw used to generate it. However, 
averaging the oracle conditional predictive density over 
the distribution of synthetic contexts recovers the original 
predictive density induced by $\rho$.

\begin{theorem}[Predictive-context identity]
\label{theo:predictive-context-identity}
Fix a target index $i$, a context index set
$A_i\subseteq\{1,\ldots,n\}\setminus\{i\}$, and a parameter density
$\rho(\theta)$. Then
\begin{equation}
\int
g_{i,A_i}^{\rho}
\left(y\mid \widetilde y_{A_i}\right)
m_{A_i}^{\rho}
\left(\widetilde y_{A_i}\right)
\,d\widetilde y_{A_i}
=
p_i^\rho(y).
\label{eq:predictive-context-identity}
\end{equation}
\end{theorem}

The identity follows directly by marginalizing the joint synthetic-data
density over the context responses. Thus,  synthetic 
contexts does not alter the marginal predictive, i.e. introduce
no approximation error under the oracle predictor.

When $\rho=\pi_{-i}$ is  the LOO posterior, 
Theorem~\ref{theo:predictive-context-identity} gives exact LOO prediction.

\begin{corollary}[LOO predictive-context identity]
\label{cor:loo-predictive-context}
Let $\rho=\pi_{-i}$. For any fixed 
$A_i\subseteq\{1,\ldots,n\}\setminus\{i\}$,
\begin{equation}
\mathbb E_{\widetilde Y_{A_i}\sim m_{A_i}^{\pi_{-i}}}
\left[
g_{i,A_i}^{\pi_{-i}}
\left(
y\mid \widetilde Y_{A_i}
\right)
\right]=
p_i^{\loo}(y).
\label{eq:loo-predictive-context-identity}
\end{equation}
\end{corollary}

\begin{remark}[Context budget]
The identity  holds for every fixed  $A_i$, so the oracle
representation does not require all $n-1$ non-held-out observations. 
Thus, the context size $M=|A_i|$ can be chosen as a learner context budget,
independently of the sample size used to construct $\rho$.
If $A_i$ is sampled independently of the synthetic responses, the
identity also holds after averaging over the distribution of $A_i$.
\end{remark}

Practical PC-LOO makes two substitutions: it generates contexts from
the available full posterior $\pi_D$ rather than $\pi_{-i}$, and replaces
the oracle predictor $g_{i,A_i}^{\pi_{-i}}$ by the external learner $F$.

\subsection{Practical PC-LOO}
For a parameter density $\rho$, synthetic context 
$\widetilde D_{A_i}$, and  learner's conditional  density 
$q_F\left(y\mid x_i,\widetilde D_{A_i}\right)$, define the 
corresponding population predictive-context density by
\begin{equation}
q_{i,A_i}^{\rho,F}(y)=
\mathbb E_{\widetilde Y_{A_i}\sim m_{A_i}^{\rho}}
\left[q_F\left(y\mid x_i,\widetilde D_{A_i}\right)\right].
\label{eq:learner-mixture}
\end{equation}
When the learner $F$ is fixed, we suppress it from the notation and
write $q_{i,A_i}^{\rho}$.

\begin{definition}[Population PC-LOO predictive density]
Using  the full posterior density
$\rho=\pi_D$,  we  define the population PC-LOO predictive
density as
\begin{equation}
q_{i,A_i}^{\mathrm{PC}}(y):=q_{i,A_i}^{\pi_D,F}(y).
\label{eq:population-pcloo} 
\end{equation}
\end{definition}

To approximate \eqref{eq:population-pcloo}, we generate $S$ independent
synthetic contexts under $\pi_D$,
$\widetilde D_{A_i}^{(1)},\ldots, \widetilde D_{A_i}^{(S)},$
and we define the PC-LOO estimator
\begin{equation}
\widehat q_{i,A_i,S}^{\mathrm{PC}}(y)
=\frac{1}{S}\sum_{s=1}^Sq_F\left(y\mid x_i,\widetilde D_{A_i}^{(s)}
\right).
\label{eq:pcloo-estimator}
\end{equation}
The pointwise and aggregate PC-LOO scores are then defined as
\begin{equation}
\widehat\ell_i^{\mathrm{PC}}
=
\log
\widehat q_{i,A_i,S}^{\mathrm{PC}}(y_i),
\qquad
\operatorname{elpd}_{\mathrm{PC}}
=
\sum_{i=1}^n
\widehat\ell_i^{\mathrm{PC}}.
\label{eq:elpd-pcloo-estimator}
\end{equation}

Equation~\eqref{eq:pcloo-estimator} introduces finite-$S$ Monte Carlo
error. Even as $S\to\infty$, however, $q_{i,A_i}^{\mathrm{PC}}$ need not
equal $p_i^{\mathrm{LOO}}$: the external learner may differ from the
oracle predictor, and contexts are generated under $\pi_D$ rather than
$\pi_{-i}$. The next section quantifies these two sources of population
approximation error.

\section{Approximation theory}
\label{sec:theory}

We quantify the population discrepancy between
$q_{i,A_i}^{\mathrm{PC}}$ and $p_i^{\loo}$ using squared Hellinger
distance,
$$
H^2(p,q)
=
\frac12\int(\sqrt{p}-\sqrt{q})^2\,dy
=
1-\int\sqrt{pq}\,dy,
$$
with $0\leq H^2(p,q)\leq1$. 
Intuitively, the error in the approximation  is due to the predictive capability 
of the learner, and  the use of the full rather than the LOO posterior. 
These two components suggest a triangle-inequality argument to control the overall
error, therefore motivating the Hellinger distance. Also, since the 
the squared Hellinger distance is upper-bounded by the total variation (TV) distance,
the results are automatically valid for TV. Proofs are provided in the supplement.

\subsection{Learner compatibility}
The discrepancy between the oracle predictor and the external learner
motivates the following definition.

\begin{definition}[Learner compatibility]
For a target index $i$, context  set $A_i$, parameter density
$\rho$, and  learner $F$, define
\begin{equation}
\C_{i,A_i}^{\rho,F}=
\mathbb E_{\widetilde Y_{A_i}\sim m_{A_i}^{\rho}}
\left[
\operatorname{KL}
\left\{
g_{i,A_i}^{\rho}(\cdot\mid\widetilde Y_{A_i})
\,
\big\|
\,
q_F(\cdot\mid x_i,\widetilde D_{A_i})
\right\}
\right].
\label{eq:compatibility}
\end{equation}
\end{definition}

Compatibily measures conditional learner error, whereas 
PC-LOO depends on the context-averaged  densities,
$ p_i^\rho = \mathbb E_{\widetilde Y_{A_i}}
\left[g_{i,A_i}^{\rho}(\cdot\mid\widetilde Y_{A_i})\right]$ and
$q_{i,A_i}^{\rho,F}=\mathbb E_{\widetilde Y_{A_i}}
\left[q_F(\cdot\mid x_i,\widetilde D_{A_i})\right].$
The following result shows that $\C_{i,A_i}^{\rho,F}$ also 
controls their discrepancy.

\begin{proposition}[Marginal compatibility bound]
\label{prop:marginal-compatibility}
For any  $i$,  $A_i$, $\rho$, and $F$,

\begin{equation}
\operatorname{KL}
\left(
p_i^\rho
\,\big\|\,
q_{i,A_i}^{\rho,F}
\right)
\leq
\mathcal \C_{i,A_i}^{\rho,F}.
\label{eq:compatibility-kl-bound}
\end{equation}
Consequently,
\begin{equation}
H^2
\left(
p_i^\rho,
q_{i,A_i}^{\rho,F}
\right)
\leq
\frac{1}{2}
\mathcal \C_{i,A_i}^{\rho,F}.
\label{eq:compatibility-hellinger-bound}
\end{equation}\end{proposition}

The first bound follows from  data-processing inequality
for  KL divergence together with the predictive-context identity; 
the second follows from $2H^2(p,q)\leq\operatorname{KL}(p\|q)$.

For practical PC-LOO, contexts are generated under the full posterior,
$\rho = \pi_D$;  define $\C_{i,A_i} = \C_{i,A_i}^{\pi_D,F}$, and 
Proposition~\ref{prop:marginal-compatibility} then
gives
\begin{equation}
H^2\left(
p_i^{\pi_D},
q_{i,A_i}^{\mathrm{PC}}
\right)
\leq
\frac{1}{2}\C_{i,A_i}.
\label{eq:practical-compatibility-bound}
\end{equation}

Thus, $\C_i$ also controls learner error under the full posterior $\pi_D$.
The remaining discrepancy from exact LOO 
arises from using $\pi_D$ instead of $\pi_{-i}$,
induced by including observation $i$ in the posterior. We quantify this
posterior perturbation next.

\subsection{Deletion influence}

Let $L_i(\theta)=p(y_i\mid x_i,\theta)$
denote the likelihood contribution of observation $i$, and let
$$
\mu_i
=\mathbb E_{\pi_{-i}}\!\left[L_i(\theta)\right]=p_i^{\loo}(y_i).
$$
\begin{definition}[Deletion influence]
We define the deletion influence of observation $i$ as
\begin{equation}
V_i
=
\frac{
\operatorname{Var}_{\pi_{-i}}
\!\left\{L_i(\theta)\right\}
}{
\mu_i^2
}.
\label{eq:deletion-influence}
\end{equation}
\end{definition}

Unlike the PSIS Pareto-$\hat{k}$, which diagnoses the tail behavior of
importance ratios, $V_i$ directly measures the posterior perturbation
induced by observation $i$. Indeed, $V_i$ can be seen as a discrepancy
between the parameter posteriors, since

$$
V_i
=
\mathbb E_{\pi_{-i}}
\left[
\left\{
\frac{\pi_D(\theta)}{\pi_{-i}(\theta)}-1
\right\}^2
\right]
=
\chi^2(\pi_D\|\pi_{-i}).
$$
The following result propagates this posterior discrepancy to the
predictive distributions.

\begin{proposition}[Predictive deletion bound]
\label{prop:predictive-deletion}
For every target $i$,
\begin{equation}
H^2\left(p_i^{\pi_{-i}},p_i^{\pi_D}\right)
\leq
\frac12 V_i.
\label{eq:predictive-deletion-bound}
\end{equation}
\end{proposition}
The result follows from  data-processing inequality for Hellinger
distance, together with $2H^2(\pi_D,\pi_{-i})
\leq \chi^2(\pi_D\|\pi_{-i})$. Together,
Propositions ~\ref{prop:marginal-compatibility} and
~\ref{prop:predictive-deletion} yields the population PC-LOO error bound.

\subsection{Practical PC-LOO  error}

The two links separating exact LOO from
practical PC-LOO are
$$
p_i^{\loo}
=
p_i^{\pi_{-i}}
\;\longrightarrow\;
p_i^{\pi_D}
\;\longrightarrow\;
q_{i,A_i}^{\mathrm{PC}},
$$
controlled respectively by deletion influence $V_i$, and
learner compatibility $\mathcal C_{i,A_i}$.
Combining these two bounds gives the following.

\begin{theorem}[PC-LOO  error bound]
\label{theo:pc-loo-bound}
For any target observation $i$ and fixed context set $A_i$,
\begin{equation}
H^2\left(
p_i^{\loo},
q_{i,A_i}^{\mathrm{PC}}
\right)
\leq
V_i+\mathcal \C_{i,A_i}.
\label{eq:pcloo-population-bound}
\end{equation}
\end{theorem}

The result follows from the triangle inequality for Hellinger
distance, and the preceding compatibility and deletion bounds.

Importantly, Equation \eqref{eq:pcloo-population-bound} controls 
predictive distributions, not the 
realized pointwise log-score error
$\log q_{i,A_i}^{\mathrm{PC}}(y_i)- \log p_i^{\loo}(y_i)$,
which may be large where either predictive density is small.

\subsection{Large-sample behavior of deletion influence}

The practical PC-LOO error decomposition also clarifies which component  should 
disappear as the dataset grows. In regular fixed-dimensional models, deleting 
one observation perturbs the posterior by an increasingly small amount as 
$n$ grows, suggesting that $V_i$ should vanish with $n$.

\begin{proposition}[Asymptotic deletion influence]
\label{prop:asymptotic-deletion}
Under standard Bernstein-von Mises regularity conditions for a
fixed-dimensional Bayesian model, together with mild local regularity
of the individual likelihood contribution,

\begin{equation}
V_i = O_p(n^{-1}).
\label{eq:deletion-rate}
\end{equation}
\end{proposition}
The result follows from Bernstein-von Mises   posterior concentration 
\citep[Sec.~10.2]{vandervaart1998}, together 
with local smoothness of the individual likelihood contribution.

Hence, the error from  replacing $\pi_{-i}$ by $\pi_D$ vanishes with sample size, 
in regular problems. In contrast, nothing  forces 
$\C_{i,A_i}$ to decrease with $n$,  so learner compatibility
may eventually dominate. If the learner is
predictively consistent for the model, however, then both
sources of  error can vanish jointly, as formalized next.

\subsection{Consistency with increasing context size}

Consider a sequence of datasets $D_n$ and
context sets $A_{i,n}\subseteq\{1,\ldots,n\}\setminus\{i\}$, with
$M_n=|A_{i,n}|\to\infty$. We say that $F$ is predictively
consistent for the model along $\{A_{i,n}\}$ if

\begin{equation}
\mathcal C_{i,A_{i,n}}^{\pi_{D_n},F}
\xrightarrow{p}0.
\label{eq:predictive-consistency}
\end{equation}
In other words, the learner's predictive converges to 
the corresponding model's predictive  with increasingly larger context.

\begin{corollary}[Consistency of PC-LOO]
\label{cor:pcloo-consistency}
Suppose that the conditions ensuring $ V_{i,n}=O_p(n^{-1})$
hold, and  $F$ is predictively consistent for the model along
$\{A_{i,n}\}$. Then

\begin{equation}
H^2\!\left(
p_{i,n}^{\mathrm{LOO}},
q_{i,A_{i,n}}^{\mathrm{PC}}
\right)
\xrightarrow{p}0.
\label{eq:pcloo-consistency}
\end{equation}
\end{corollary}

We stress that consistency does not
require $M_n=n-1$; any context sequence for which 
$\mathcal C_{i,A_{i,n}}^{\pi_D,F}\to0$ is sufficient, and
whether a  predictive learner satisfies this property is 
learner- and model-dependent.

\section{Model comparison}
\label{sec:model-comparison}

The previous results quantify how closely PC-LOO approximates exact LOO
 for a single model. For model comparison,
however, differential rather than absolute approximation
error is decisive.
Consider two models $\mathcal M_0$ and $\mathcal M_1$, and write
$$
L_k
=
\operatorname{elpd}_{\mathrm{LOO}}(\mathcal M_k),
\quad
\widehat L_k
=
\operatorname{elpd}_{\mathrm{PC}}(\mathcal M_k),
\quad k\in\{0,1\}.
$$
Define the PC-LOO approximation error for model $k$ as
$$
E_k
=
\widehat L_k-L_k.
$$
The exact and PC-LOO model differences are
$$
\Delta
=
L_1-L_0,
\qquad
\widehat\Delta
=
\widehat L_1-\widehat L_0.
$$

Since $$\widehat\Delta-\Delta =E_1-E_0,$$ 
approximation errors that are common to both models cancel in
the model difference. Thus, PC-LOO can preserve the LOO model
comparison even when neither model's absolute LOO score is estimated
with high accuracy.

\begin{proposition}[Model-ranking preservation]
\label{prop:model-rank-preservation}
If $\Delta\neq0$, then
\begin{equation}
|E_1-E_0|<|\Delta|
\quad\Longrightarrow\quad
\operatorname{sign}(\widehat\Delta)
=
\operatorname{sign}(\Delta).
\label{eq:ranking-condition}
\end{equation}
\end{proposition}

Hence, large absolute errors can preserve the ranking when they are
similar across models, while small errors can reverse it when
$|\Delta|$ is small. 

\section{Experiments}
\label{sec:experiments}

We organize the experiments around the theoretical decomposition above,
rather than around an extensive benchmark suite. 
Code and scripts reproducing all experiments, figures, and tables are
available in the following 
\href{https://anonymous.4open.science/r/pc-loo-1D42}{repository}.

\subsection{Experimental setup}

The theory applies to a generic predictive learner $F$, while practical
PC-LOO is instantiated with TabICL. We use controlled Gaussian regression
experiments  where the oracle conditional predictor, exact LOO
distribution, and quantities appearing in the theory are available.

Real-data experiments use Stack Loss ($n=21$), South German Credit 
($n=1,000$), California Housing
($n=20{,}640$), and YearPredictionMSD ($n=463{,}715$). Stack Loss, 
and the Credit data provide a small and moderate problem where exact LOO
can be evaluated exhaustively; California Housing is used to study 
context budget and practical pointwise accuracy; and YearPredictionMSD
isolates large-$n$ behavior. For California Housing, computation is 
restricted to a subset of targets stratified by $V_i$ to include 
observations ranging from weakly to highly influential.
Methods are evaluated on the same selected targets, with
exact LOO as the reference, and PSIS-LOO as another 
practical baseline. In particular, the  Credit data
is used to illustrate PC-LOO on TabICL and a binary
classification problem.
Additional details  are given in the supplementary material.

\subsection{Testing the bound theory}

We first test the PC-LOO bound theory in a controlled Gaussian regression
setting, where the required quantities can be computed directly.
Figure~\ref{fig:population-theory} evaluates the predictive-context
identity, the compatibility and deletion bounds, and their combination.
The oracle approximation converges to $p_i^\rho$ for all considered
$\rho$, while the observed predictive discrepancies remain below their
corresponding theoretical bounds. These results support the proposed
decomposition of PC-LOO error into learner compatibility and posterior
deletion influence.

\begin{figure}[t]
    \centering
    \includegraphics[width=\linewidth]{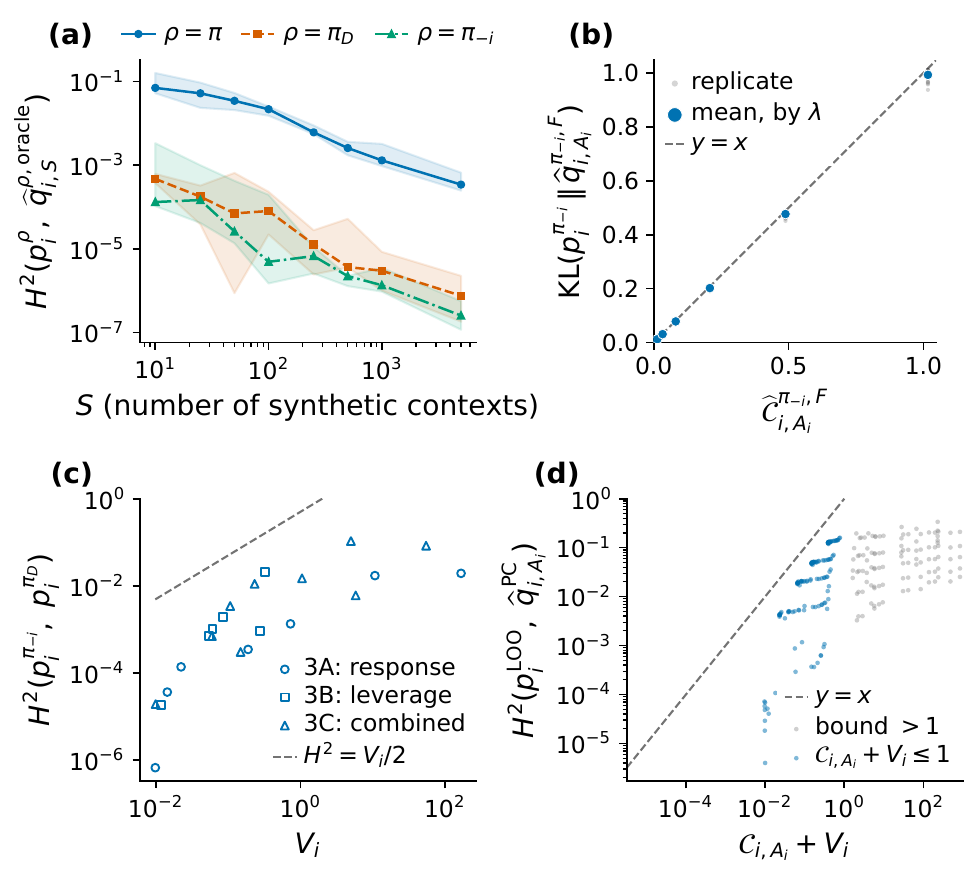}
    \caption{
    Validation of the population theory in controlled Gaussian regression.
    (a) The finite-$S$ oracle approximation converges to $p_i^\rho$ for
    $\rho\in\{\pi,\pi_D,\pi_{-i}\}$, supporting the predictive-context
    identity.
    (b) The marginal predictive KL discrepancy is bounded by the estimated
    learner compatibility under $\rho=\pi_{-i}$.
    (c) The predictive shift
    $H^2(p_i^{\pi_{-i}},p_i^{\pi_D})$ remains below $V_i/2$ across response,
    leverage, and combined perturbations.
    (d) The total PC-LOO discrepancy
    $H^2(p_i^{\mathrm{LOO}},q_{i,A_i}^{\mathrm{PC}})$ is bounded by
    $\mathcal C_{i,A_i}+V_i$; gray points indicate the trivial regime
    $\mathcal C_{i,A_i}+V_i>1$.
    }
    \label{fig:population-theory}
\end{figure}

\subsection{Context budget and large-$n$ behavior}

The predictive-context identity holds for any fixed context set, but
practical accuracy and cost depend on the learner's context budget
$M$. Figure~\ref{fig:context-budget} shows on California Housing that
increasing $M$ generally reduces pointwise PC-LOO error and learner
compatibility, with diminishing gains as computation increases.

In the sequence, we hold the context budget fixed and increase the
sample size on YearPredictionMSD. Figure~\ref{fig:large-n} 
shows that deletion influence decreases approximately at the predicted 
$n^{-1}$ rate, whereas pointwise PC-LOO error does not decrease comparably.
The compatibility diagnostic on a small fixed set of targets is consistent 
with this remaining error being learner-limited rather than deletion-limited.

\begin{figure}[t]
    \centering
    \includegraphics[width=\linewidth]{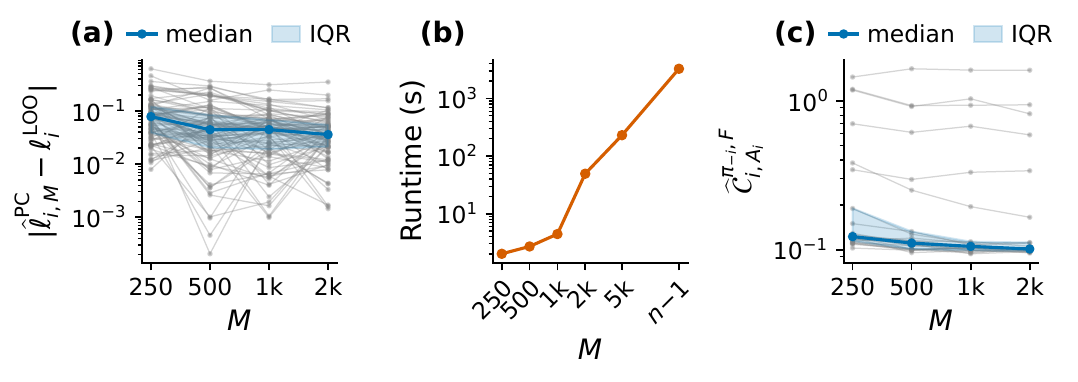}
    \caption{
    Context-budget behavior on California Housing.
    Increasing $M$ generally reduces pointwise PC-LOO error and the
    learner-compatibility, while computational cost rises
    substantially. Panel~(c) uses LOO-posterior contexts.
    }
    \label{fig:context-budget}
\end{figure}

\begin{figure}[t]
    \centering
    \includegraphics[width=\linewidth]{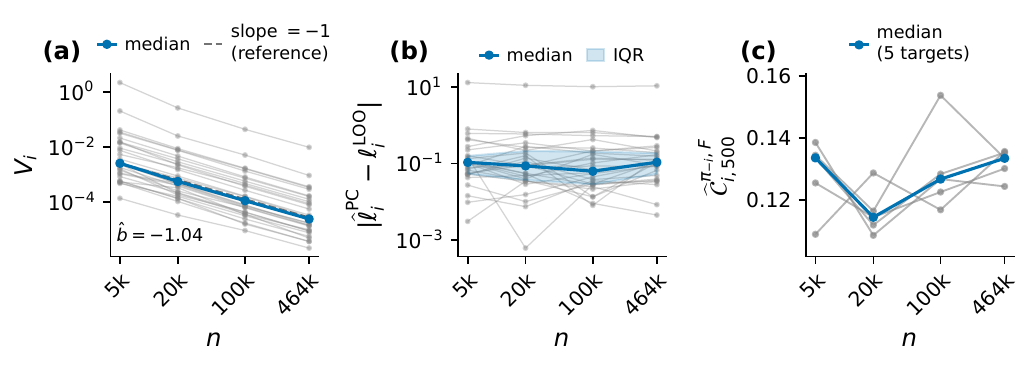}
    \caption{
    Large-$n$ behavior on YearPredictionMSD with fixed context budget.
    (a) Deletion influence decreases approximately at the predicted
    $n^{-1}$ rate.
    (b) Pointwise PC-LOO error does not decrease at the same rate.
    (c) Compatibility on a small fixed set of targets remains roughly
    stable across $n$.
    }
    \label{fig:large-n}
\end{figure}

\subsection{Model comparison}

Here, we examine if PC-LOO preserves differences between competing
models. In controlled Gaussian experiments with varying signal strength,
Figure~\ref{fig:model-comparison} compares exact
$\Delta_{\mathrm{LOO}}$ with oracle and TabICL PC-LOO estimates, and
reports ranking agreement across quartiles of standardized separation
$z=|\Delta_{\mathrm{LOO}}|/\mathrm{SE}(\Delta_{\mathrm{LOO}})$.
The oracle closely tracks exact LOO, while TabICL introduces additional
scatter; in both cases, ranking agreement increases as model separation
grows, consistent with Proposition~\ref{prop:model-rank-preservation}.

\begin{figure}[t]
    \centering
    \includegraphics[width=\linewidth]{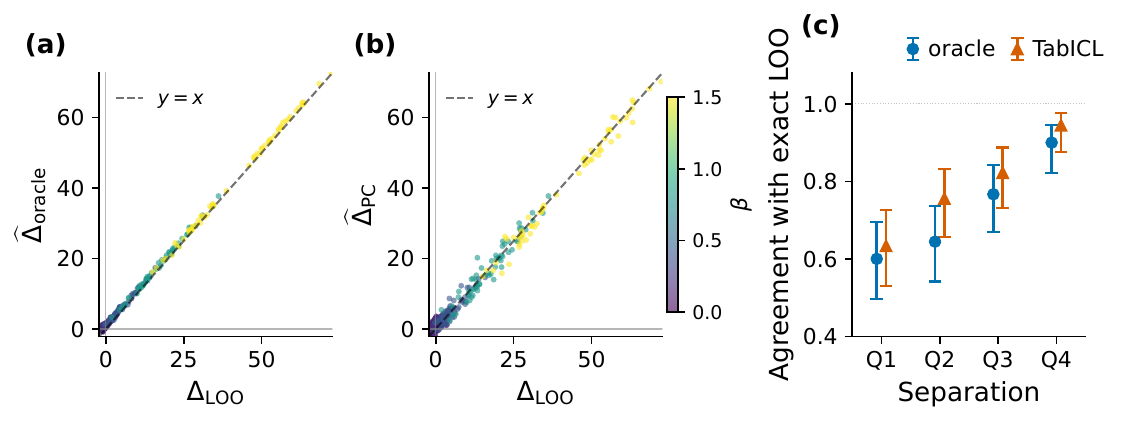}
    \caption{
    Model comparison in controlled Gaussian regression.
    (a) Oracle predictive-context estimates of the model difference
    closely track exact $\Delta_{\mathrm{LOO}}$.
    (b) TabICL PC-LOO introduces additional approximation error.
    (c) Agreement with the exact LOO ranking increases across quartiles
    of standardized separation
    $z=|\Delta_{\mathrm{LOO}}|/\mathrm{SE}(\Delta_{\mathrm{LOO}})$.
    }
    \label{fig:model-comparison}
\end{figure}
\subsection{Practical accuracy on real data}

We next assess practical pointwise accuracy against exact LOO.
Stack Loss and the South German credit data permit exhaustive evaluation
over all observations, while California Housing uses the stratified
target subset described above.
Table~\ref{tab:practical-accuracy} reports absolute pointwise
log-predictive-density errors for PC-LOO and  PSIS-LOO.
PC-LOO errors remain moderate on these examples, although 
PSIS is  more accurate in these regular regression
settings.

\begin{table}[t]
\centering
\caption{Pointwise approximation error from PC-LOO and PSIS-LOO
relative to exact LOO.
Reported are mean absolute error (MAE) and root mean squared error
(RMSE).}
\label{tab:practical-accuracy}
\small
\setlength{\tabcolsep}{3pt}
\begin{tabular}{llrrrr}
\toprule
& & \multicolumn{2}{c}{PC-LOO} &
    \multicolumn{2}{c}{PSIS-LOO} \\
Data & Model & MAE & RMSE & MAE & RMSE \\
\midrule
Stack & $\mathcal{M}_0$ & .301 & .649 & .004 & .005 \\
      & $\mathcal{M}_1$ & .608 & 1.312 & .014 & .048 \\
      & $\mathcal{M}_2$ & .473 & .824 & .013 & .024 \\
Credit & $\mathcal{M}$ & .089 & .154 & $.002$ & $.004$ \\
California & $\mathcal{M}$ & .064 & .090 & $2.3{\times}10^{-4}$ & $3.6{\times}10^{-4}$ \\
\bottomrule
\end{tabular}
\end{table}

\subsection{Preliminary reliability diagnostics}

The population bounds do not provide an observation-specific diagnostic
for realized PC-LOO error. We therefore compare deletion influence
$V_i$ with the PSIS Pareto shape estimate $\widehat k_i$, and examine
how pointwise approximation error varies with $\widehat k_i$.
Figure~\ref{fig:reliability} shows that $V_i$ and $\widehat k_i$ are
positively associated, while larger $\widehat k_i$ also tends to
correspond to larger pointwise error. However, $\widehat k_i$ is not a
complete diagnostic for PC-LOO, and no universal PC-LOO reliability
threshold is established here. Additional experiments on the reliability 
are provided in the supplement.

\begin{figure}[t]
    \centering
    \includegraphics[width=\linewidth]{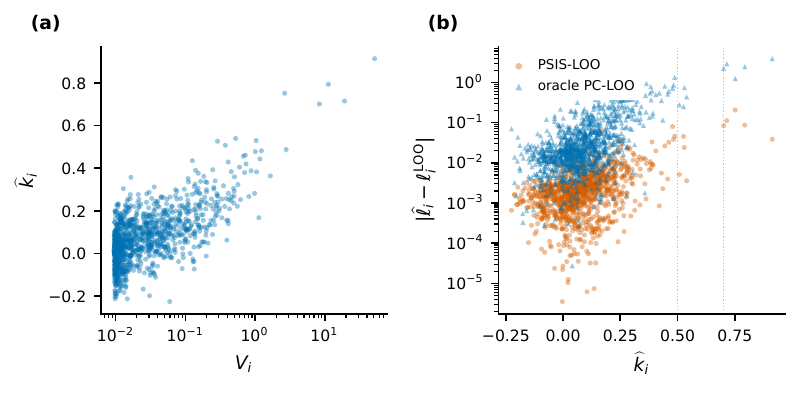}
    \caption{
    Preliminary reliability analysis.
    (a) Deletion influence $V_i$ versus PSIS Pareto shape estimate
    $\widehat k_i$.
    (b) Absolute pointwise LPD error versus $\widehat k_i$ for PSIS-LOO
    and oracle PC-LOO. Larger $\widehat k_i$ is associated with more
    difficult observations, but does not provide a calibrated PC-LOO
    reliability threshold.
    }
    \label{fig:reliability}
\end{figure}

\section{Conclusion}
\label{sec:conclusion}

PC-LOO provides a predictive-context interface from a fitted Bayesian model
to an external predictive learner. The oracle identity shows that
posterior-predictive contexts can represent the Bayesian predictive
distribution exactly after averaging over contexts; population approximation
enters when the oracle conditional predictor is replaced by a learned
predictor and the LOO posterior by the full posterior. The resulting
decomposition into learner compatibility and deletion influence makes
these two sources of error explicit.

Deletion influence decreases approximately at the predicted $n^{-1}$
rate in the regular large-sample experiment, while practical PC-LOO
error need not decrease at the same rate when the learner context budget
is fixed.  For model comparison, the relevant quantity is differential instead of absolute
approximation error. The experiments also show that learner compatibility can remain the
dominant source of error even when deletion influence becomes small.
This suggests a strategy for improving PC-LOO,
for example through learner selection, ensembles or mixtures of
predictive learners, or lightweight adaptation to
posterior-predictive contexts. A Reliability diagnostic that jointly
reflect posterior deletion influence and learner mismatch,
to complement $V_i$ and Pareto-$\widehat{k}_i$, is needed and left for future work.

Finally, the framework is more general than the present TabICL
experiments. First, since PC-LOO  approximates the full LOO 
posterior distribution, it is also directly applicable to other scoring rules.  
In addition,  as better and more diverse predictive learners become available,
including learners for dependent, structured, or non-tabular data, the same
predictive-context construction can be reused without changing the
underlying theory. This also suggests extensions beyond pointwise LOO,
such as joint leave-group prediction, where synthetic contexts could
support more realistic prediction tasks involving blocks, clusters, or
other dependent groups. 



\bibliographystyle{apalike}
\bibliography{references}

\onecolumn
\clearpage\onecolumn\appendix
\aistatstitle{Supplementary Material}
\section{Proofs}
\label{app:proofs}

Throughout, all  distributions are assumed to admit densities
with respect to a common dominating measure, and all divergences
appearing below are assumed to be finite.

\subsection{Proof of Theorem~\ref{theo:predictive-context-identity}}
\begin{proof}
By definition,
$$
g_{i,A_i}^{\rho}(y\mid\widetilde y_{A_i})
=
\frac{
m_{i,A_i}^{\rho}(y,\widetilde y_{A_i})
}{
m_{A_i}^{\rho}(\widetilde y_{A_i})
}.
$$
Therefore,
\begin{align*}
\int
g_{i,A_i}^{\rho}(y\mid\widetilde y_{A_i})
m_{A_i}^{\rho}(\widetilde y_{A_i})
\,d\widetilde y_{A_i}
&\qquad=
\int
m_{i,A_i}^{\rho}(y,\widetilde y_{A_i})
\,d\widetilde y_{A_i}
\\
&\qquad=
\int
p(y\mid x_i,\theta)
\left[
\int
\prod_{j\in A_i}
p(\widetilde y_j\mid x_j,\theta)
\,d\widetilde y_{A_i}
\right]
\rho(\theta)\,d\theta
\\
&\qquad=
\int
p(y\mid x_i,\theta)\rho(\theta)\,d\theta
=
p_i^\rho(y),
\end{align*}
where the third equality follows by Tonelli's theorem and the inner
integral equals one.
\end{proof}

\subsection{Proof of Corollary~\ref{cor:loo-predictive-context}}
\begin{proof}
Apply Theorem~\ref{theo:predictive-context-identity} with
$\rho=\pi_{-i}$.
\end{proof}

\subsection{Proof of Proposition~\ref{prop:marginal-compatibility}}
\begin{proof}
Define two joint densities on
$(\widetilde Y_{A_i},Y)$ by
$$
P(\widetilde y_{A_i},y)
=
m_{A_i}^{\rho}(\widetilde y_{A_i})
g_{i,A_i}^{\rho}(y\mid\widetilde y_{A_i}) \quad \text{ and } \quad 
Q(\widetilde y_{A_i},y)
=
m_{A_i}^{\rho}(\widetilde y_{A_i})
q_F(y\mid x_i,\widetilde D_{A_i}).
$$
Because $P$ and $Q$ have the same marginal distribution for
$\widetilde Y_{A_i}$,
\begin{align*}
\KL(P\|Q)
&=
\mathbb E_{\widetilde Y_{A_i}\sim m_{A_i}^{\rho}}
\left[
\KL\left\{
g_{i,A_i}^{\rho}(\cdot\mid\widetilde Y_{A_i})
\|
q_F(\cdot\mid x_i,\widetilde D_{A_i})
\right\}
\right]
\\
&=
\mathcal C_{i,A_i}^{\rho,F}.
\end{align*}
Marginalizing $\widetilde Y_{A_i}$ maps $P$ and $Q$ respectively to
$p_i^\rho$ and $q_{i,A_i}^{\rho,F}$. Hence the data-processing
inequality for KL divergence gives
$$
\KL\left(
p_i^\rho
\|
q_{i,A_i}^{\rho,F}
\right)
\leq
\mathcal C_{i,A_i}^{\rho,F}.
$$
Finally, using $H^2(P,Q)\leq \KL(P\|Q)/2$ gives
$$
H^2\left(
p_i^\rho,
q_{i,A_i}^{\rho,F}
\right)
\leq
\frac12\mathcal C_{i,A_i}^{\rho,F}.
$$
\end{proof}

\subsection{Proof of Proposition~\ref{prop:predictive-deletion}}

\begin{proof}

First, we prove that $V_i$ is a $\chi^2$ divergence. 

Bayes' rule gives
$$
\frac{\pi_D(\theta)}{\pi_{-i}(\theta)}
=
\frac{L_i(\theta)}{\mu_i},
\qquad
\mu_i
=
\mathbb E_{\pi_{-i}}[L_i(\theta)].
$$
Therefore,
\begin{align*}
\chi^2(\pi_D\|\pi_{-i})
&=
\mathbb E_{\pi_{-i}}
\left[
\left\{
\frac{\pi_D(\theta)}{\pi_{-i}(\theta)}-1
\right\}^2
\right]
\\
&=
\mathbb E_{\pi_{-i}}
\left[
\left\{
\frac{L_i(\theta)}{\mu_i}-1
\right\}^2
\right]
\\
&=
\frac{
\Var_{\pi_{-i}}\{L_i(\theta)\}
}{
\mu_i^2
}
=
V_i.
\end{align*}

Let
$$
K_i(\theta,dy)
=
p(y\mid x_i,\theta)\,dy
$$
denote the predictive Markov kernel. Applying $K_i$ to $\pi_{-i}$
and $\pi_D$ gives respectively $p_i^{\pi_{-i}}$ and
$p_i^{\pi_D}$. Since squared Hellinger distance is an
$f$-divergence, the data-processing inequality yields
$$
H^2\left(
p_i^{\pi_{-i}},
p_i^{\pi_D}
\right)
\leq
H^2(\pi_{-i},\pi_D).
$$
Moreover,
$$
H^2(\pi_{-i},\pi_D)
\leq
\frac12\chi^2(\pi_D\|\pi_{-i})
=
\frac12V_i.
$$
Combining the two inequalities proves the result.
\end{proof}

\subsection{Proof of Theorem~\ref{theo:pc-loo-bound}}
\begin{proof}
By the triangle inequality for Hellinger distance,
\begin{align*}
H\left(
p_i^{\loo},
q_{i,A_i}^{\mathrm{PC}}
\right)
&\leq
H\left(
p_i^{\pi_{-i}},
p_i^{\pi_D}
\right)
+
H\left(
p_i^{\pi_D},
q_{i,A_i}^{\mathrm{PC}}
\right)
\\
&\leq
\sqrt{\frac{V_i}{2}}
+
\sqrt{\frac{\mathcal C_{i,A_i}}{2}},
\end{align*}
where the second inequality follows from
Propositions~\ref{prop:predictive-deletion} and
\ref{prop:marginal-compatibility}. Squaring gives
$$
H^2\left(
p_i^{\loo},
q_{i,A_i}^{\mathrm{PC}}
\right)
\leq
\frac12
\left(
\sqrt{V_i}
+
\sqrt{\mathcal C_{i,A_i}}
\right)^2.
$$
Finally,
$2\sqrt{V_i\mathcal C_{i,A_i}}
\leq V_i+\mathcal C_{i,A_i}$,
so
$$
H^2\left(
p_i^{\loo},
q_{i,A_i}^{\mathrm{PC}}
\right)
\leq
V_i+\mathcal C_{i,A_i}.
$$
\end{proof}

\subsection{Proof of Proposition~\ref{prop:asymptotic-deletion}}
\begin{proof}
Suppose the parameter dimension is fixed and that the LOO posterior
satisfies Bernstein--von Mises concentration around the true parameter
$\theta_0$, with posterior second moments satisfying
\begin{equation}
\mathbb E_{\pi_{-i}}
\left[
\|\theta-\theta_0\|^2
\right]
=
O_p(n^{-1}).
\label{eq:bvm-second-moment}
\end{equation}
Such behavior holds under the usual regularity conditions, together
with the corresponding moment convergence; see, e.g.,
\citet[Sec.~10.2]{vandervaart1998}.

Assume additionally that, in a neighborhood of $\theta_0$,
$L_i(\theta)$ is locally Lipschitz with random Lipschitz constant
$B_i=O_p(1)$, and that posterior contributions outside this
neighborhood are negligible at the same order. Then
$$
|L_i(\theta)-L_i(\theta_0)|
\leq
B_i\|\theta-\theta_0\|
$$
on this neighborhood. Hence
\begin{align*}
\Var_{\pi_{-i}}\{L_i(\theta)\}
&\leq
\mathbb E_{\pi_{-i}}
\left[
\{L_i(\theta)-L_i(\theta_0)\}^2
\right]
\\
&=
O_p(1)
\,
\mathbb E_{\pi_{-i}}
\left[
\|\theta-\theta_0\|^2
\right]
\\
&=
O_p(n^{-1}).
\end{align*}
If, in addition,
$\mu_i^{-1}=O_p(1)$, then
$$
V_i
=
\frac{
\Var_{\pi_{-i}}\{L_i(\theta)\}
}{
\mu_i^2
}
=
O_p(n^{-1}).
$$
\end{proof}
The result is pointwise in $i$; it does not imply a uniform
$O_p(n^{-1})$ bound over increasingly extreme observations.

\subsection{Proof of Proposition~\ref{prop:model-rank-preservation}}
\begin{proof}
Since $ \widehat\Delta=\Delta+(E_1-E_0).
$
If $\Delta>0$ and $|E_1-E_0|<|\Delta|$, then $\widehat\Delta>\Delta-|\Delta|=0.$
Similarly, if $\Delta<0$, $\widehat\Delta<\Delta+|\Delta|=0.$
Therefore
$$
\operatorname{sign}(\widehat\Delta)
=
\operatorname{sign}(\Delta).
$$
\end{proof}

\subsection{Proof of Corollary~\ref{cor:pcloo-consistency}}

\begin{proof}
By Theorem~\ref{theo:pc-loo-bound}, applied to the sequence
$(D_n,A_{i,n})$,
$$
H^2\!\left(
p_{i,n}^{\mathrm{LOO}},
q_{i,A_{i,n}}^{\mathrm{PC}}
\right)
\leq
V_{i,n}
+
\mathcal C_{i,A_{i,n}}^{\pi_{D_n},F}.
$$
Under the stated regularity conditions,
$V_{i,n}=O_p(n^{-1})$, and therefore $V_{i,n}\xrightarrow{p}0$.
By predictive consistency of $F$ for the model along
$\{A_{i,n}\}$, we have that 
$
\mathcal C_{i,A_{i,n}}^{\pi_{D_n},F}
\xrightarrow{p}0.
$
Hence,
$$
V_{i,n}
+
\mathcal C_{i,A_{i,n}}^{\pi_{D_n},F}
\xrightarrow{p}0,
\text{ and consequently }
H^2\!\left(
p_{i,n}^{\mathrm{LOO}},
q_{i,A_{i,n}}^{\mathrm{PC}}
\right)
\xrightarrow{p}0.
$$
\end{proof}

\section{Additional Experimental Details}
\label{sec:supp-experiments}

This section provides additional details for the experiments reported in
the main text. The experiments have two complementary purposes. The
controlled Gaussian experiments isolate the mechanisms appearing in the
theory, for which the relevant predictive distributions and discrepancy
measures can be evaluated directly. The real-data experiments instead
evaluate practical PC-LOO with TabICL and study the effects of context
size, sample size, and computational approximation.

\subsection{General PC-LOO implementation}
\label{sec:supp-pcloo-implementation}

For each target observation $i$, practical PC-LOO proceeds as follows:

\begin{enumerate}
    \item Fit the Bayesian model once to the full dataset $D$ and obtain
    the posterior $\pi_D(\theta)$.

    \item For $s=1,\ldots,S$, draw
    $\theta^{(s)}\sim\pi_D(\theta)$ and generate a synthetic context
    $$
    \widetilde D_{A_i}^{(s)}
    =
    \{(x_j,\widetilde Y_j^{(s)}):j\in A_i\},
    \qquad
    \widetilde Y_j^{(s)}
    \sim p(\cdot\mid x_j,\theta^{(s)}),
    $$
    where $A_i\subseteq\{1,\ldots,n\}\setminus\{i\}$ and
    $M=|A_i|$.

    \item Supply $\widetilde D_{A_i}^{(s)}$ as context to the predictive
    learner and $x_i$ as the query, obtaining
    $q_F(\cdot\mid x_i,\widetilde D_{A_i}^{(s)})$.

    \item Average the predictive densities across contexts,
    $$
    \widehat q_{i,A_i,S}^{\mathrm{PC}}(y)
    =
    \frac{1}{S}\sum_{s=1}^S
    q_F(y\mid x_i,\widetilde D_{A_i}^{(s)}),
    $$
    and evaluate
    $$
    \widehat\ell_i^{\mathrm{PC}}
    =
    \log\widehat q_{i,A_i,S}^{\mathrm{PC}}(y_i).
    $$
\end{enumerate}

The same $\theta^{(s)}$ is used for all observations within a synthetic
context, preserving dependence induced by posterior uncertainty.
Unless stated otherwise, synthetic contexts are generated from
$\pi_D$.

\subsection{Computational structure and complexity}
\label{sec:supp-complexity}

PC-LOO requires one Bayesian model fit, followed by predictive-learner
evaluations for each target and synthetic context. Let
$C_F(M)$ denote the cost of evaluating the learner with a context of
size $M$. Evaluating all $n$ targets using $S$ synthetic contexts has
nominal cost $O\!\left(nS\,C_F(M)\right),$
in addition to the single Bayesian fit and posterior-predictive
simulation. The corresponding learner-side memory requirement depends
on the implementation of $F$ and its batching strategy.

The computation is highly parallel. Conditional on the fitted Bayesian
model, different targets $i$ and synthetic contexts $s$ can be evaluated
independently. When the context indices $A_i$ are fixed across
replicates, the covariates $X_{A_i}$ are also unchanged: only the
synthetic responses $\widetilde Y_{A_i}^{(s)}$ vary. This structure
allows contexts to be batched across posterior-predictive replicates
and potentially permits reuse of computations involving the fixed
covariates. Thus, the nominal $nS$ learner evaluations need not imply
$nS$ sequential evaluations.

The context budget $M$ provides an additional computational control.
The predictive-context identity is exact for any fixed $M$, so reducing
$M$ introduces no oracle subsampling bias. In practical PC-LOO, the
resulting statistical-computational trade-off instead arises through
learner compatibility: smaller contexts reduce computation but may
provide less information to the external learner.
\subsection{Controlled Gaussian experiments}
\label{sec:supp-gaussian}

The controlled experiments use
$
x_j \stackrel{\mathrm{iid}}{\sim} N(0,1),
\quad
y_j = \alpha + \beta x_j + \varepsilon_j,
\quad
\varepsilon_j \stackrel{\mathrm{iid}}{\sim} N(0,\sigma^2),
$
with $n=50$, $\alpha=0$, $\beta=1$, and $\sigma^2=1$.  
The fitted Bayesian linear model uses an intercept and slope, with
$
y\mid\boldsymbol\beta,\sigma^2
\sim N(X\boldsymbol\beta,\sigma^2I),
\quad
\boldsymbol\beta\mid\sigma^2
\sim N(0,10\sigma^2I),
\quad
\sigma^2\sim\operatorname{InvGamma}(2,2).
$
Posterior inference is analytic under the conjugate
Normal-Inverse-Gamma model, and Exact LOO quantities, with Student-t,  
are obtained by analytic leave-one-out updating. Also, unless stated otherwise,
the synthetic context contains all non-target observations, so $M=n-1=49$.

\subsubsection{Predictive-context identity}
\label{sec:supp-identity}

We numerically verify the predictive-context identity under
$\rho\in\{\pi,\pi_D,\pi_{-i}\}$, prior, full posterior and leave-data-out
posterior respectively. For each $\rho$, three targets are
selected according to leverage (central, moderate, and highest), and
synthetic contexts are generated using the full $M=49$ non-target
observations. Prediction is performed with the exact Student-t 
predictive density, so the only approximation is Monte Carlo
averaging over contexts. We use $S\in\{10,25,50,100,250,500,1000,5000\}$
and compare the resulting predictive with the exact Student-$t$
predictive density $p_i^\rho$ using squared Hellinger distance.

\subsubsection{Learner compatibility}
\label{sec:supp-compatibility}

To control learner mismatch explicitly, we replace the oracle student-t 
predictive density $g$ by
$
q_\lambda=(1-\lambda)g+\lambda r,
$
where $r$ is Gaussian with mean shifted by two oracle predictive standard
deviations and standard deviation $1.5$ times that of $g$. We use
$
\lambda\in\{0,.02,.05,.10,.20,.40,.70,1\},
$
with $S=500$ synthetic contexts and $R=10$ independent Monte Carlo
repetitions for each setting.
Contexts are generated from $\pi_{-i}$ with $M=49$. For each setting, we
compute the conditional compatibility
$\mathcal C_{i,A_i}^{\pi_{-i},F}$ and the KL divergence between the
corresponding context-averaged predictive density and the exact LOO
predictive density. Numerical integrals are evaluated on a common
4000-point grid.

\subsubsection{Deletion influence}
\label{sec:supp-deletion}

To vary deletion influence systematically, observation $i=49$ is modified
under three regimes: (i) response contamination,
$y_i=\alpha+\beta x_i+k\sigma$, with
$k\in\{0,.5,1,2,3,4,6\}$; (ii) leverage contamination,
$x_i\in\{0,1,2,3,5,8\}$, with $y_i$ drawn from the true model; and
(iii) combined contamination, with
$x_i\in\{0,3,6\}$ and $k\in\{0,2,4\}$.
For each condition, $V_i$ is estimated using $20{,}000$ draws from the
exact LOO posterior $\pi_{-i}$. The predictive discrepancy is approximated
using $S=1000$ synthetic contexts under each of $\pi_D$ and $\pi_{-i}$,
averaging the corresponding exact Student-$t$ predictive densities, 
and is compared with the theoretical bound $V_i/2$.

\subsubsection{Combined approximation bound}
\label{sec:supp-combined}

We jointly vary learner mismatch and deletion influence using
$
\lambda\in\{0,.05,.15,.30,.60\},
\quad
k\in\{0,1,2,4,6\},
$
where $\lambda$ controls the predictive-density perturbation defined
above and $k$ controls response contamination,
$y_i=\alpha+\beta x_i+k\sigma$, for $i=49$.
For each of the 25 settings, we use $R=10$ independent repetitions and
$S=500$ synthetic contexts with $M=49$.
The PC-LOO predictive density is computed from $\pi_D$ contexts, while
$\mathcal C_i$ is estimated using $\pi_{-i}$ contexts. We compare
$H^2(p_i^{\mathrm{LOO}},\widehat q_i^{\mathrm{PC}})$ with
$\mathcal C_i+V_i$, where $V_i$ is estimated using $20{,}000$ draws from
$\pi_{-i}$.

\subsubsection{Finite-context Monte Carlo convergence}
\label{sec:supp-mc}

We study Monte Carlo error as the number of synthetic contexts varies over
$
S\in\{2,5,10,20,50,100,250,500,1000\}.
$
For the central-leverage target, contexts are generated from $\pi_D$ with
$M=49$. We consider both the exact Student-$t$ conditional predictive
density and the perturbed density with $\lambda=.3$. Each setting is
repeated $200$ times, using context samples across values of $S$.
We report the RMSE of the estimated predictive density at $y_i$. The
reference is the exact $p_i^{\pi_D}(y_i)$ in the Student-$t$ case and an
$S=5000$ Monte Carlo estimate in the perturbed case. Log-log regressions
of RMSE on $S$ give slopes $-0.493$ and $-0.490$, respectively, close to
the Monte Carlo rate $S^{-1/2}$, shown in  Figure~\ref{fig:supp-mc}.

\begin{figure}[t]
    \centering
    \includegraphics[width=0.7\linewidth]{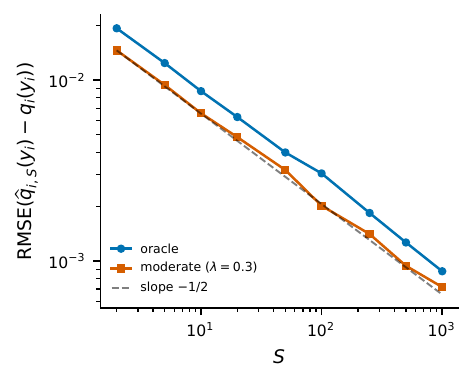}
    \caption{\textbf{Finite-context Monte Carlo convergence.}
    RMSE of the estimated predictive density at $y_i$ as a function of
    the number of synthetic contexts $S$. Results are shown for the exact
    Student-$t$ conditional predictive density and the perturbed density
    with $\lambda=.3$. The fitted log-log slopes are $-0.493$ and
    $-0.490$, respectively, close to the $S^{-1/2}$ Monte Carlo rate.}
    \label{fig:supp-mc}
\end{figure}

\subsection{California Housing experiments}
\label{sec:supp-california}

We use the California Housing dataset with all $n=20{,}640$ observations
and eight covariates. Covariates are standardized using their full-data
means and standard deviations; the response is left untransformed. We fit
the Gaussian linear model with an intercept and all eight covariates,
using the same Normal-Inverse-Gamma prior as in the controlled
experiments. Posterior inference and LOO prediction are analytic.

\subsubsection{Context size}
\label{sec:supp-california-context}

We evaluate PC-LOO at
$
M\in\{250,500,1000,2000\}
$
using the same $100$ target observations at every $M$. Targets are selected
by stratifying all observations into quintiles of $V_i$ and sampling $20$
from each quintile. For each target and $M$, we use $S=20$ synthetic
contexts generated from $\pi_D$, with a new subset of $M$ non-target
observations sampled without replacement for each context. Accuracy is
measured by
$
\left|
\widehat\ell_i^{\mathrm{PC}}-\ell_i^{\mathrm{LOO}}
\right|,
$
where $\ell_i^{\mathrm{LOO}}$ is computed analytically.

\subsubsection{Runtime with increasing context size}
\label{sec:supp-california-runtime}
Runtime is evaluated separately on eight target observations using
$
M\in\{250,500,1000,2000,5000,20639\}
$
and $S=10$ synthetic contexts per target. Contexts are generated from
$\pi_D$ and resampled for each replicate. Runtime is the mean wall-clock
time per target for TabICL prediction and scoring across all $S$
contexts; posterior fitting, exact LOO computation, and calculation of
$V_i$ are excluded.

\subsubsection{Compatibility with increasing context size}
\label{sec:supp-california-compatibility}

Compatibility is evaluated on 25 targets, with five sampled from each
$V_i$ quintile, using
$
M\in\{250,500,1000,2000\}
$
and $S=30$ synthetic contexts per target and $M$. Context subsets are
resampled without replacement for each replicate and generated from the
LOO posterior $\pi_{-i}$. We estimate
$
\widehat{\mathcal C}_{i,A_i}^{\pi_{-i},F}
=
\frac1S\sum_{s=1}^S
\mathrm{KL}\!\left(
g_{i,A_i}^{\pi_{-i}}(\cdot\mid\widetilde D_{A_i}^{(s)})
\,\|\,q_F(\cdot\mid x_i,\widetilde D_{A_i}^{(s)})
\right),
$
where the first density is the exact Student-$t$ predictive density and
$F$ is TabICL.

\subsubsection{Fixed versus resampled context subsets}
\label{sec:supp-fixed-resampled}

We compare two context-selection schemes at $M\in\{1000,2000\}$ using
the same 100 targets as the context-size experiment and $S=20$. In the
fixed scheme, one subset $A_i$ is sampled for each target and $M$ and
reused across all $S$ synthetic contexts. In the resampled scheme, a new
subset $A_i^{(s)}$ is sampled for each context. Both schemes use
$\pi_D$. Table~\ref{tab:fixed-resampled} reports pointwise PC-LOO error and mean
runtime. The two schemes have similar predictive accuracy, while fixing
the context subset reduces computation.

\begin{table}[H]
\centering
\caption{\textbf{Fixed versus resampled context subsets.}
PC-LOO absolute error and mean runtime per target.}
\label{tab:fixed-resampled}
\begin{tabular}{llrrrr}
\toprule
$M$ & Context & Median & Mean & RMSE & Time (s)\\
\midrule
1000 & Fixed     & .0415 & .0545 & .0825 & 6.19\\
1000 & Resampled & .0445 & .0535 & .0731 & 9.24\\
2000 & Fixed     & .0371 & .0481 & .0638 & 94.41\\
2000 & Resampled & .0359 & .0468 & .0659 & 101.91\\
\bottomrule
\end{tabular}
\end{table}

\subsection{YearPredictionMSD experiments}
\label{sec:supp-year}

We use the canonical training split of YearPredictionMSD, containing
$n=463{,}715$ observations and 90 covariates. Covariates and the response
are standardized using statistics from the full training split. We fit a
Gaussian linear model with an intercept and all 90 covariates, using the
same Normal-Inverse-Gamma prior as above.

We consider nested sample sizes
$
n\in\{5000,20000,100000,463715\},
$
obtained as prefixes of a single random permutation of the training data.
Thirty targets, six from each $V_i$ quintile, are selected at $n=5000$
and tracked across all sample sizes. A subset of five targets, one per
quintile, is used for the compatibility experiment.

\subsubsection{Deletion-influence scaling}
\label{sec:supp-year-deletion}

For each sample size, $V_i$ is estimated for the same 30 targets using
$20{,}000$ draws from the corresponding exact LOO posterior. To estimate
the scaling rate, we regress $\log V_i$ on $\log n$ separately for each
target and average the resulting 30 slopes. The mean slope is $-1.044$,
with 95\% confidence interval $[-1.079,-1.009]$. Figure~\ref{fig:large-n}
reports the median $V_i$ across targets at each $n$, together with an
$n^{-1}$ reference line.

\subsubsection{PC-LOO error with increasing sample size}
\label{sec:supp-year-error}

We evaluate PC-LOO on the same 30 targets at each sample size, using
$M=500$ and $S=10$. Synthetic contexts are generated from the
corresponding full posterior $\pi_{D_n}$. For each target and $n$, one
subset of $M$ non-target observations is sampled and reused across the
$S$ synthetic contexts. Accuracy is measured by
$
\left|
\widehat\ell_i^{\mathrm{PC}}-\ell_i^{\mathrm{LOO}}
\right|,
$
where $\ell_i^{\mathrm{LOO}}$ is computed analytically.
Figure~\ref{fig:large-n} reports the distribution of this error across
the 30 targets as $n$ increases.

\subsubsection{Compatibility with increasing sample size}
\label{sec:supp-year-compatibility}

Compatibility is evaluated on five of the 30 targets, one from each
$V_i$ quintile, using $M=500$ and $S=10$. At each sample size,
synthetic contexts are generated from the corresponding LOO posterior
$\pi_{-i}$. We estimate
$
\widehat{\mathcal C}_{i,500}^{\pi_{-i},F}
=
\frac{1}{S}\sum_{s=1}^S
\mathrm{KL}\!\left(
g_{i}^{\pi_{-i}}(\cdot\mid\widetilde D_{A_i}^{(s)})
\,\|\,q_F(\cdot\mid x_i,\widetilde D_{A_i}^{(s)})
\right),
$
where $g_i^{\pi_{-i}}$ is the exact Student-$t$ predictive density and
$F$ is TabICL.

\subsection{Practical real-data accuracy}
\label{sec:supp-real-accuracy}

We use the Stack Loss data with $n=21$ observations and three standardized
covariates. We consider an intercept-only model $M_0$, a model $M_1$
containing Air Flow, and a model $M_2$ containing all three covariates.
All models use the same conjugate prior as above. PC-LOO is evaluated for
all observations using the full $M=20$ context and $S=250$ synthetic
contexts generated from $\pi_D$. Exact LOO is computed analytically.
For comparison, PSIS-LOO uses $4000$ draws from $\pi_D$.

The South German Credit data contains $n=1000$ observations, $20$ covariates
and a binary response. We fit a Bayesian logistic regression with standardized
covariates and weakly informative Gaussian priors. PC-LOO is evaluated
for all observations using TabICL with context sizes
$M\in\{250,500,999\}$,  $S=20$ synthetic contexts generated from
$\pi_D$; and PSIS-LOO uses $4000$ draws from $\pi_D$. The improved
accuracy of PC-LOO for larger context is shown in Figure 
\ref{fig:credit-classification}.

For California Housing, we use the model and preprocessing described in
Section~\ref{sec:supp-california}. We evaluate 100 targets, with 20
selected from each $V_i$ quintile. PC-LOO uses $M=500$, $S=30$, and one
fixed context subset per target, with synthetic responses generated from
$\pi_D$. PSIS-LOO uses $2000$ posterior draws. Exact analytic LOO is used
as the reference for both methods.

For these datasets, pointwise accuracy is summarized by
$
\operatorname{MAE}
=
\frac1m\sum_{i=1}^m
|\widehat\ell_i-\ell_i^{\mathrm{LOO}}|,
\quad
\operatorname{RMSE}
=
\left\{
\frac1m\sum_{i=1}^m
(\widehat\ell_i-\ell_i^{\mathrm{LOO}})^2
\right\}^{1/2}.
$

\begin{figure}[H]
    \centering
    \includegraphics[width=\linewidth]{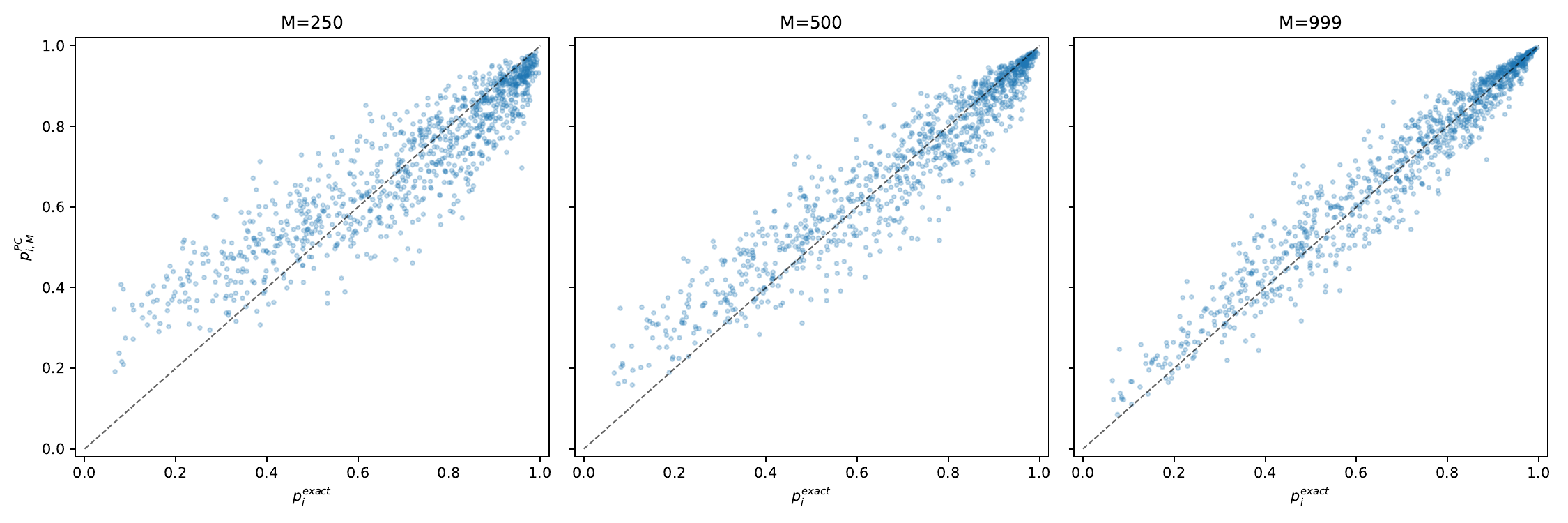}
    \caption{
    PC-LOO versus exact LOO predictive probabilities for the South German
    Credit data at context sizes $M\in\{250,500,999\}$.
    Agreement improves with context size: pointwise log-score correlation
    increases from $0.908$ to $0.968$, while RMSE decreases from $0.250$
    to $0.155$ and MAE from $0.156$ to $0.089$.
    The dashed line denotes exact agreement.
    }
    \label{fig:credit-classification}
\end{figure}

\subsection{Posterior-influence diagnostics}
\label{sec:supp-psis}

We compare the deletion-influence measure $V_i$ with the standard
PSIS Pareto diagnostic $\widehat k_i$ using the 22 conditions from the
deletion-influence experiment. Diagnostics are computed for all 50
observations in each condition, giving 1100 pointwise evaluations,
and for PSIS, we use 2000 draws from $\pi_D$. The corresponding $V_i$ values 
are estimated using $20{,}000$ draws from the exact LOO posterior.
We additionally compare $\widehat k_i$ with the absolute pointwise error
of PSIS-LOO and predictive-context LOO based on the exact Student-$t$
conditional predictive density. The latter uses $S=100$ synthetic
contexts generated from $\pi_D$.

\subsection{Exploratory PC-LOO reliability diagnostic}
\label{sec:supp-pc-diagnostic}

A practical question is whether the two sources of PC-LOO approximation
error can be reflected in a single observation-specific diagnostic.
We explored a Pareto-tail construction for the PC-LOO method. 

For posterior draws $\theta_s\sim\pi_D$, define the usual LOO
importance ratio
$$
W_s^{\mathrm{post}}
=
\frac{1}{p(y_i\mid x_i,\theta_s)}.
$$For a synthetic context $\widetilde D_{A_i}^{(s)}$ generated from
$\pi_D$, let
$g_{i,A_i}^{\pi_D}(\cdot\mid\widetilde D_{A_i}^{(s)})$
denote the exact conditional predictive density and
$q_F(\cdot\mid x_i,\widetilde D_{A_i}^{(s)})$
the learner predictive density. Drawing
$
Y_s\sim
q_F(\cdot\mid x_i,\widetilde D_{A_i}^{(s)}),
$
define
$$
W_s^{F}
=
\frac{
g_{i,A_i}^{\pi_D}
(Y_s\mid\widetilde D_{A_i}^{(s)})
}{
q_F(Y_s\mid x_i,\widetilde D_{A_i}^{(s)})
}.
$$
We then form
$$
W_s^{\mathrm{PC}}
=
W_s^{\mathrm{post}}W_s^{F},
\qquad
\log W_s^{\mathrm{PC}}
=
\log W_s^{\mathrm{post}}+\log W_s^{F},
$$
and estimate the corresponding Pareto shape parameter
$\widehat k_i^{\mathrm{PC}}$ using the standard PSIS tail fit.

We evaluate this construction in the controlled Gaussian experiment
with $n=50$, $M=49$, and $S=2000$. We use 30 independently generated
datasets for each combination
$
\lambda\in\{0,.05,.15,.30,.60\},
\quad
k\in\{0,1,2,4,6\},
$
where $\lambda$ controls learner mismatch and $k$ controls response
contamination. This gives 750 target-level evaluations.

Across this experiment, $\widehat k_i^{\mathrm{PC}}$ is strongly associated 
with absolute pointwise PC-LOO error, see Figure \ref{fig:kpc-diagnostic}.
This result suggests that tail behavior of the combined ratio may contain useful
information about difficult PC-LOO cases. This experiment should nevertheless 
be interpreted as exploratory rather than as establishing a calibrated
diagnostic. In particular, small $\widehat k_i^{\mathrm{PC}}$ does not 
in general guarantee small PC-LOO error. While the construction provides
a useful direction for PC-LOO reliability assessment, its calibration
and interpretation remain open.

\begin{figure}[H]
    \centering
    \includegraphics[width=0.65\linewidth]{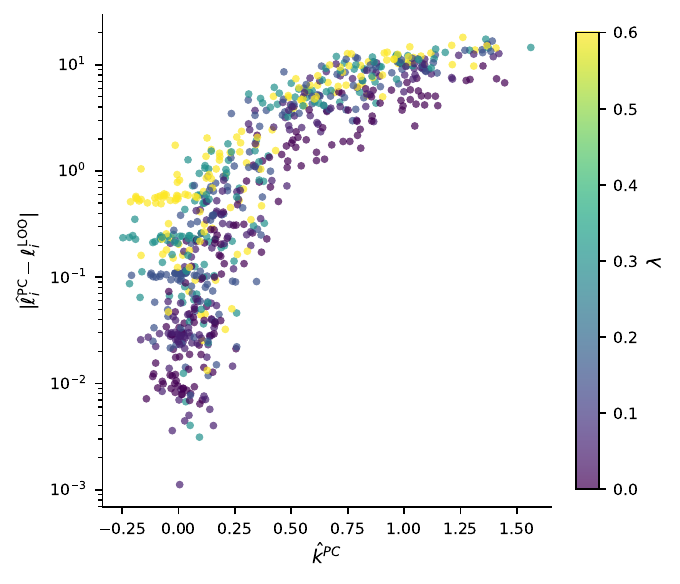}
    \caption{\textbf{Exploratory PC-LOO tail diagnostic.}
    Absolute pointwise PC-LOO error as a function of the estimated
    Pareto shape parameter $\widehat{k}_i^{\mathrm{PC}}$ for the combined
    ratio $W^{\mathrm{PC}}=W^{\mathrm{post}}W^{F}$.
    The strong association is suggestive of diagnostic value, but
    $\widehat{k}_i^{\mathrm{PC}}$ should not be interpreted as a
    calibrated reliability diagnostic.}
    \label{fig:kpc-diagnostic}
\end{figure}

\section{Computational details}
\label{sec:supp-computation}

\paragraph{TabICL predictive densities.}
TabICL returns predictive quantiles rather than an explicit density.
We evaluate 199 quantiles with probability levels from $.005$ to $.995$
and construct a piecewise-constant predictive density between adjacent
quantiles. The resulting density is numerically normalized before
evaluation. For PC-LOO, predictive densities are averaged across
synthetic contexts before taking the logarithm.

\paragraph{TabICL computation and parallelism.}
For TabICL version 2.1.1, learner evaluation becomes the dominant practical cost as
the context size $M$ grows. Our implementation processes synthetic
contexts separately and therefore does not fully exploit their shared
structure. For a fixed target and context set, all $S$ replicates share
the same covariates $X_{A_i}$ and differ only in the synthetic responses
$\widetilde Y_{A_i}^{(s)}$. Joint batching of these response replicates
could therefore reuse computations involving the covariates and reduce
the effective cost across $S$. Most experiments were run on a  
NVIDIA GTX 1080 Ti GPU. 

\paragraph{Numerical divergences.}
KL and squared Hellinger divergences are evaluated numerically on a
common grid of 4000 points covering the relevant predictive
distributions. Both densities are normalized on this grid before
integration, which is performed using the trapezoidal rule.

\paragraph{Implementation.}
All Bayesian linear-regression experiments use analytic
Normal-Inverse-Gamma posterior updating and analytic Student-$t$
predictive distributions.  Runtime measurements include TabICL prediction 
and scoring across all synthetic contexts for a target, but exclude Bayesian
posterior fitting, exact LOO calculations, and computation of $V_i$.

\paragraph{PSIS-LOO.}
PSIS-LOO is computed using the implementation in \texttt{ArviZ}, applying
\texttt{arviz.psislw} to the log importance ratios.
We use independent draws from the exact conjugate posterior and set
\texttt{reff=1.0}. The Pareto-smoothed weights are then used to compute
the pointwise PSIS-LOO predictive densities. We use $4000$ posterior
draws for Stack Loss and $2000$ for California Housing and for the
controlled posterior-influence diagnostics.

\end{document}